\documentclass[10pt]{article}
\pdfoutput=1
\usepackage{fullpage}
\usepackage[bookmarks,colorlinks=true,linkcolor=blue,urlcolor=blue,citecolor=blue,anchorcolor=green,pdfusetitle,breaklinks=true]{hyperref}
\usepackage{amssymb}
\usepackage{amsmath,array}
\newcolumntype{P}[1]{>{\centering\arraybackslash}p{#1}}
\newcolumntype{M}[1]{>{\centering\arraybackslash}m{#1}}
\usepackage{amsthm}
\usepackage{tcolorbox}
\usepackage{graphicx,color,colordvi}
\usepackage{bbm}
\usepackage{cleveref}
\usepackage{stmaryrd}
\usepackage[utf8]{inputenc}
\usepackage[blocks]{authblk}
\usepackage{dsfont}
\usepackage{mathtools}

\usepackage{centernot}
\usepackage{booktabs}

\usepackage[backend=bibtex8,firstinits=true,doi=false,isbn=false,url=false,maxbibnames=5]{biblatex}
\renewbibmacro{in:}{}
\newbibmacro{string+doi}[1]{\iffieldundef{doi}{#1}{\href{https://dx.doi.org/\thefield{doi}}{#1}}}
\DeclareFieldFormat{title}{\usebibmacro{string+doi}{\mkbibemph{#1}}}
\DeclareFieldFormat[article]{title}{\usebibmacro{string+doi}{\mkbibquote{#1}}}
\DeclareFieldFormat[incollection]{title}{\usebibmacro{string+doi}{\mkbibquote{#1}}}                   
\DeclareFieldFormat[inproceedings]{title}{\usebibmacro{string+doi}{\mkbibquote{#1}}}

\def\rl{\rangle \langle}
\def\openone{\leavevmode\hbox{\small1\kern-3.8pt\normalsize1}}
\def\ch{{\cal H}}
\def\cm{{\cal M}}

\def\cn{{\cal N}}

\def\11{\mathbb{I}}

\def\sln{\succeq_{\operatorname{l.n.}}}

\newcommand{\esln}[1]{\succeq_{\operatorname{l.n.}}^{{#1}}}

\def\rel{\succeq_{\operatorname{rel}}}
\newcommand{\erel}[1]{\succeq_{\operatorname{rel}}^{{#1}}}

\newtheorem{definition}{Definition}[section]

\newtheorem{lemma}[definition]{Lemma}

\newtheorem{theorem}[definition]{Theorem}
\newtheorem{corollary}[definition]{Corollary}
\newtheorem{conjecture}[definition]{Conjecture}

\newcommand{\supp}{\mathop{\rm supp}\nolimits}

\newcommand{\tr}{\mathop{\rm Tr}\nolimits}

\newcommand{\im}{\mathop{\rm Im}\nolimits}

\usepackage{pst-node}
\usepackage{tikz-cd}

\newcommand{\cA}{{\cal A}}
\newcommand{\cB}{{\cal B}}
\newcommand{\cD}{{\cal D}}
\newcommand{\cE}{{\cal E}}

\newcommand{\cN}{{\cal N}}
\newcommand{\cT}{{\cal T}}
\newcommand{\cH}{{\cal H}}

\newcommand{\cR}{{\cal R}}

\newcommand{\cM}{{\mathcal{M}}}

\newcommand{\coh}[1]{Q^{(1)}\!\left(#1\right)}
\newcommand{\priv}[1]{P^{(1)}\!\left(#1\right)}

\newcommand{\Qss}[1]{Q_{ss}\!\left(#1\right)}

\usepackage{graphicx}
\usepackage{setspace}
\usepackage{verbatim}
\usepackage{subfig}

\numberwithin{equation}{section}

\DeclareRobustCommand\openone{\leavevmode\hbox{\small1\normalsize\kern-.33em1}}

\newcommand{\id}{{\rm{id}}}
\newcommand{\be}{\begin{equation}}
	\newcommand{\ee}{\end{equation}}
\newcommand{\bea}{\begin{eqnarray}}
	\newcommand{\eea}{\end{eqnarray}}
\newcommand{\beas}{\begin{eqnarray*}}
	\newcommand{\eeas}{\end{eqnarray*}}

\title{Bounding quantum capacities via partial orders and complementarity}

\begin{document}

\author[1,2]{Christoph Hirche\thanks{\texttt{christoph.hirche@gmail.com}}}
\affil[1]{Zentrum Mathematik, Technical University of Munich, 85748 Garching, Germany}
\affil[2]{Centre for Quantum Technologies, National University of Singapore, Singapore}
\author[3]{Felix Leditzky\thanks{\texttt{leditzky@illinois.edu}}}
\affil[3]{Department of Mathematics \& IQUIST, University of Illinois at Urbana-Champaign, Urbana, IL 61801, USA}

\date{}

\maketitle

\begin{abstract}
Quantum capacities are fundamental quantities that are notoriously hard to compute and can exhibit surprising properties such as superadditivity. Thus, a vast amount of literature is devoted to finding tight and computable bounds on these capacities. We add a new viewpoint by giving operationally motivated bounds on several capacities, including the quantum capacity and private capacity of a quantum channel and the one-way distillable entanglement and private key of a quantum state. These bounds are generally phrased in terms of capacity quantities involving the complementary channel or state. As a tool to obtain these bounds, we discuss partial orders on quantum channels and states, such as the less noisy and the more capable order. Our bounds help to further understand the interplay between different capacities, as they give operational limitations on superadditivity and the difference between capacities in terms of the information-theoretic properties of the complementary channel or state. They can also be used as a new approach towards numerically bounding capacities, as discussed with some examples. 
\end{abstract}

\section{Overview and main results} 

Capacities give the optimal rate at which a certain information theoretic task can be achieved. As such, they play a fundamental role in understanding the capabilities afforded by a specific resource such as a quantum channel or a quantum state. Specific tasks of interest include for example public or private information transmission over a channel, and the distillation of maximally entangled or private states. 
In many cases we even know of mathematical formulas that exactly determine these capacities.
Those could already be the end of our journey; however, to really understand or even numerically evaluate these quantities still remains an extremely challenging task. Two typical questions are as follows. First, we know from operational arguments that the rate at which we can transmit private classical information over a quantum channel is never smaller than the rate at which we can send quantum information over the same channel. But it is often unknown how much more exactly of the former can be sent. Second, in both of these examples the capacity is given by a regularized formula, meaning it has to be evaluated on $n$ copies of the channel in the limit of $n$ going to infinity. This makes numerical evaluation generally intractable. It is again easy to see that the regularized quantity can never be smaller than the single-copy version it is based on, but it is a priori unclear how much bigger the regularized quantity can become. 

Due to these challenges, a significant part of the quantum information literature strives to find better bounds on quantum channel capacities that help us to narrow down their numerical value, and hence give a better understanding of their information-theoretic capabilities. A small collection of recent results on upper bounds on capacities includes for example~\cite{sutter2017approximate,WXD16,leditzky2017useful, leditzky2018approaches, leditzky2018quantum, wang2018semidefinite, fanizza2020quantum, fawzi2021hierarchy,wang2019optimizing, fang2021geometric}.
Naturally, a main focus in this area has been to find approximations of capacities in terms of upper bounds that can be easily evaluated numerically. However, it can often be difficult to assign any operational understanding to these bounds. In this work we address the latter point by finding bounds on capacities that have an operational interpretation themselves, ideally phrased in terms of capacities.
These bounds may shed further light on the information-theoretic structures that allow for phenomena such as superadditivity. 

An important concept in this work will be that of complementarity. It is well known that one can think of any quantum channel $\cN$ as an isometric embedding into a larger (tensor product) space, followed by discarding the auxiliary system which is usually referred to as the environment. The complement of that channel, denoted $\cN^c$, is obtained by keeping the environment while discarding the original output system. Information-theoretically, the complementary channel models the leakage of information to the environment. Note that while the complement of a given channel is not unique, all choices are information-theoretically equivalent. The concept of complementarity can also be applied to mixed bipartite states shared between two parties, say, Alice and Bob. Purifying such a shared state and discarding Bob's system results in a complementary state quantifying the correlations between Alice and the environment.

As a starting point for our discussion, consider the class of degradable channels \cite{devetak2005capacity}. Those are channels for which the receiver can apply another channel $\cD$ to simulate the complementary channel, i.e., $\cN^c=\cD\circ\cN$. Intuitively, this implies that the channel $\cN$ should never be worse at transmitting information than $\cN^c$. As a consequence of degradability, the quantum capacity $Q(\cN)$ and private capacity $P(\cN)$ of a degradable channel $\cN$ simplify~\cite{devetak2005capacity,smith2008private} (see Sec.~\ref{sec:definitions} for a more detailed discussion of these capacities):  %
\begin{align}
P(\cN) = Q(\cN) = Q^{(1)}(\cN) = P^{(1)}(\cN), \label{Eq:degraded}
\end{align}
where the channel's coherent information $Q^{(1)}(\cN)$ and private information $P^{(1)}(\cN)$ are the corresponding non-regularized, single-copy quantities, defined in \eqref{eq:coherent-information} and \eqref{eq:private-information} below, respectively.

Equation~\eqref{Eq:degraded} for degradable channels hints at the fact that the relationship between a channel and its complement determines properties of their capacities. 
Later, Watanabe~\cite{watanabe2012private} made this idea more precise by translating the classical concept of less noisy and more capable channels~\cite{shannon1958note} to the quantum setting. Both of these classes had previously proven useful in classical information theory, but Watanabe realized that they gain new meaning when applied to a quantum channel and its complement. Namely, we call a channel regularized less noisy when the private capacity of its complement is zero, $P(\cN^c)=0$, and regularized more capable when its complement's quantum capacity is zero, $Q(\cN^c)=0$. Note that regularized less noisy implies regularized more capable by the well-known capacity inequalities $0\leq Q(\cN)\leq P(\cN)$ valid for any quantum channel $\cN$. Moreover, a degradable channel $\cN$ satisfies $P(\cN^c)=0$, since the existence of the degrading map makes it impossible for the sender to transmit private information to the environment. Hence, degradability implies both regularized less noisy and more capable. Watanabe \cite{watanabe2012private} showed that (a) relaxing degradability to regularized less noisy is still sufficient for \eqref{Eq:degraded} to hold; (b) regularized more capable still implies  $P(\cN) = Q(\cN).$

Naturally, it is desirable to see what we can learn from these results for general channels. To this end, Sutter et al.~introduced the concept of approximately degradable channels~\cite{sutter2017approximate}, showing that the relations in Equation~\eqref{Eq:degraded} still hold approximately when a channel is close to being degradable in a suitable sense. This idea led to some of the best capacity bounds available which are even efficiently computable as the optimal approximation constant is given by a convex optimization problem. The recent work~\cite{HRS20} introduced approximately less noisy and more capable classes, leading to potentially tighter bounds, however at the cost of generally losing the efficient computability. Here, we remedy this disadvantage by showing that the approach can be used to give bounds with operational meaning that extend on the previously achieved results. 

We will now discuss the main results of this work, while referring to the later sections for technical definitions, statements and proofs. In particular, the technical sections include new results on connections between classes of channels and partial orders that might be of independent interest beyond the capacity bounds presented here. 

Our main results regarding quantum channels and their capacities are discussed in Sec.~\ref{sec:channels}.
As a warm-up to the structure of our results, we give bounds on the classical capacity $C(\cN)$ and the entanglement-assisted classical capacity $C_E(\cN)$ in Theorem~\ref{Thm:ClassicalBounds},  
\begin{alignat}{2}
Q(\cN) &\leq C(\cN)  &&\leq Q(\cN) + C(\cN^c) \label{eq:classical-bound}\\
2Q^{(1)}(\cN) &\leq C_E(\cN) &&\leq  2Q^{(1)}(\cN) + C_E(\cN^c).
\end{alignat}
Note that in contrast to the other capacity formulas discussed here, $C_E(\cN)$ does not require regularization \cite{bennett1999entanglement} and can be efficiently computed~\cite{fawzi2018efficient} (see Sec.~\ref{sec:definitions} for a more detailed discussion). 

Next, we focus on the private and quantum capacity. In Corollary~\ref{Cor:CapBounds} we extend the results in~\cite{HRS20}, showing that the quantum capacity of the complementary channel limits how different the private and quantum capacity of the channel can be:
\begin{alignat}{3}
Q^{(1)}(\cn) &\leq P^{(1)}(\cn) &&\leq Q^{(1)}(\cn) &&+ Q^{(1)}(\cn^c), \label{eq:QP-inf-bounds}\\
Q(\cn) &\leq P(\cn) &&\leq Q(\cn) &&+ Q(\cn^c).\label{eq:QP-cap-bounds}
\end{alignat}
Similarly, the entanglement-assisted private information $P_E(\cN^c)$ (defined in \eqref{eq:EA-priv} below) limits the increase due to regularization,
\begin{align}
	Q^{(1)}(\cn) &\leq Q(\cn) \leq Q^{(1)}(\cn) + P_E(\cN^c), \label{eq:Q-reg-bound}\\
	P^{(1)}(\cn) &\leq P(\cn) \leq P^{(1)}(\cn) + Q(\cN^c) + P_E(\cN^c).\label{eq:P-reg-bound}
\end{align}
The entanglement-assisted private information $P_E(\cN)$ was proven in~\cite{qi2018entanglement} to equal the entanglement-assisted private capacity of degradable channels. This extends a result in~\cite{cross2017uniform} which, translated to our notation, states that $Q^{(1)}(\cn) = Q(\cn)$ if $P_E(\cN^c)=0$. 
While the condition $P_E(\cN^c)=0$ is referred to as `informationally degradable' in \cite{cross2017uniform}, we refer to this property as `fully quantum less noisy' in this work. 

The above bounds give an operationally meaningful, quantitative version of the results by Watanabe \cite{watanabe2012private}. Furthermore, they make the intuition precise that the properties of the complementary channel of a general channel limit the possibility of having superadditivity or a higher private capacity than quantum capacity in a fundamental way. 
In Sec.~\ref{sec:examples} we discuss how our bounds can be used to obtain numerical bounds on the private capacity. 
We also identify a potentially new class of zero-private-capacity channels that we call `bi-PPT' channels, consisting of quantum channels $\cN$ such that both $\cN$ and $\cN^c$ are PPT \cite{HorodeckiHorodeckiEA00}; such channels have vanishing private capacity by our bound \eqref{eq:QP-cap-bounds}, and may lead to an observation of superactivation of private capacity.
In numerical studies we found approximate examples of such bi-PPT channels with small (but provably positive) private capacity.

Section~\ref{Sec:States} then slightly changes focus from investigating channels to discussing quantum states. Approximate degradable quantum states were defined in~\cite{leditzky2017useful} and used therein to give bounds on the one-way distillable entanglement $D_{\rightarrow}(\rho_{AB})$. Additionally, we consider here the one-way distillable private key $K_{\rightarrow}(\rho_{AB})$. We define new partial orders based on these two quantities, which lead us to results similar to the channel setting. First, we define the complementary state $\rho_{AB}^c$ of a state $\rho_{AB}$ as $\rho_{AB}^c := \rho_{AE} =\tr_B\Psi_{ABE}$ where $\Psi_{ABE}$ is a purification of $\rho_{AB}$. We then show in Theorem~\ref{Thm:distEntKey} that the one-way distillable entanglement of the complementary state limits the difference between distillable key and entanglement, 
\begin{alignat}{2}
D^{(1)}_{\rightarrow}(\rho_{AB}) &\leq K^{(1)}_{\rightarrow}(\rho_{AB}) &&\leq D^{(1)}_{\rightarrow}(\rho_{AB}) + D^{(1)}_{\rightarrow}(\rho_{AB}^c) \label{eq:DK-inf-bounds}\\
D_{\rightarrow}(\rho_{AB}) &\leq K_{\rightarrow}(\rho_{AB}) &&\leq D_{\rightarrow}(\rho_{AB}) + D_{\rightarrow}(\rho_{AB}^c).\label{eq:DK-cap-bounds}
\end{alignat}
Similarly, the complement state's one-way distillable key limits the increase due to regularization, see Theorem~\ref{Thm:distEntReg} and Corollary~\ref{Cor:DistKeyReg}, 
\begin{alignat}{2}
D^{(1)}_{\rightarrow}(\rho_{AB}) &\leq D_{\rightarrow}(\rho_{AB}) &&\leq D^{(1)}_{\rightarrow}(\rho_{AB}) + K_{\rightarrow}(\rho_{AB}^c) \label{eq:D-reg-bound}\\
K^{(1)}_{\rightarrow}(\rho_{AB}) &\leq K_{\rightarrow}(\rho_{AB}) &&\leq K^{(1)}_{\rightarrow}(\rho_{AB}) + K_{\rightarrow}(\rho_{AB}^c) + D_{\rightarrow}(\rho_{AB}^c). \label{eq:K-reg-bound}
\end{alignat}
Note the formal equivalence between \cref{eq:QP-inf-bounds,eq:QP-cap-bounds,eq:Q-reg-bound,eq:P-reg-bound} and \cref{eq:DK-inf-bounds,eq:DK-cap-bounds,eq:D-reg-bound,eq:K-reg-bound}.
Together, these results show that a similar intuition as for channels also holds for quantum states: the possibility of extracting certain resources from the complementary state determines properties of the capacities of the state itself. 

Finally, in Section~\ref{Sec:sidechannel} we discuss symmetric side-channel assisted capacities and how superactivation is directly related to the question whether the sets of degradable and regularized less noisy channels are actually different.
That is, we show the implication
\begin{align}
P(\cdot) \text{ can be superactivated} \quad\Rightarrow\quad \mathrm{DEG} \subsetneq \mathrm{LN}_\infty,
\end{align} 
where $\mathrm{DEG}$ and $\mathrm{LN}_\infty$ denote the classes of degradable and regularized less noisy channels, respectively.

We end by discussing some open problems in Section~\ref{Sec:Outlook}, intended to inspire further research in this direction.  
The appendices contain some additional proofs and observations supporting the main text.

\section{Partial orders on channels and their implications} \label{sec:channels}

\subsection{Definitions and notation}

In this paper classical and quantum systems are denoted by capital letters. Generally, $A$, $B$, $E$ denote quantum systems and $U$, $T$, $X$ denote classical systems.
A (classical or quantum) system $R$ is associated with a finite-dimensional Hilbert space $\cH_R$.
A quantum state $\rho_R$ on $R$ is a positive semidefinite linear operator with unit trace acting on $\cH_R$.
A state $\rho_R$ of rank $1$ is called pure, and we may choose a normalized vector $|\psi\rangle_R\in\cH_R$ satisfying $\rho_R=|\psi\rangle\langle\psi|_R$.
Otherwise, $\rho_R$ is called a mixed state. 
By the spectral theorem, every mixed state can be written as a convex combination of pure states.
For a pure state $|\phi\rangle$ we often use the shorthand $\phi \equiv |\phi\rangle\langle\phi|$.
For a classical system $X$ there is a distinguished orthonormal basis $\lbrace |x\rangle\rbrace_{x=1}^{\dim \cH_X}$ of $\cH_X$ diagonalizing every state on $X$.
For a quantum state $\rho_A$ we denote by $H(A)_\rho = -\tr\rho_A\log\rho_A$ the von Neumann entropy.
For a bipartite state $\rho_{AB}$ acting on the tensor product space $\cH_A\otimes \cH_B$, we denote by $I(A:B)_\rho = H(A)+H(B)-H(AB)$ the mutual information.
For a tripartite state $\rho_{ABC}$ acting on the tensor product space $\cH_A\otimes \cH_B\otimes \cH_C$, we denote by $I(A:B|C)_\rho = H(AC)+H(BC)-H(ABC)-H(C)$ the conditional mutual information.

A quantum channel $\cN\colon A\to B$ is a linear completely positive and trace-preserving map from the space of linear operators on $\cH_A$ to those on $\cH_B$.
For every quantum channel $\cN\colon A \to B$ we can choose an auxiliary space $\cH_E$, usually called the environment, and an isometry $V\colon \cH_A\to \cH_B\otimes \cH_E$, usually called a Stinespring isometry, such that $\cN(X_A) = \tr_E (V X_A V^\dagger)$.
A channel isometry gives rise to the so-called complementary channel $\cN^c\colon A\to E$ modeling the loss of information to the environment, defined as $\cN^c(X_A) = \tr_B (V X_A V^\dagger)$.
Letting $|\gamma\rangle_{AA'} = \sum_{i=1}^{\dim\cH_A} |i\rangle_A\otimes |i\rangle_{A'}$ be an unnormalized maximally entangled state defined with respect to an orthonormal basis $\lbrace |i\rangle_A\rbrace_{i}$ of $\cH_A$, the Choi operator of $\cN$ is defined as $\tau_{AB} = (\id_A\otimes \cN)(\gamma_{AA'})$.
A quantum channel $\cN\colon A\to B$ with complementary channel $\cN^c\colon A\to E$ is called degradable if there exists another channel $\cD\colon B\to E$ satisfying $\cN^c = \cD\circ \cN$ \cite{devetak2005capacity}.
A quantum channel is called anti-degradable if its complementary channel is degradable.
In analogy to channels and their complementary channels, we can also define the related concept of a complementary state of a bipartite state $\rho_{AB}$: considering a purification $|\psi\rangle_{ABE}$ satisfying $\rho_{AB} = \tr_E \psi_{ABE}$, the complementary state is defined as $\rho_{AB}^c \coloneqq \rho_{AE} = \tr_B \psi_{ABE}$ \cite{leditzky2017useful}.
Degradability and antidegradability of states are defined similarly as for channels.

\subsection{Partial orders and channel capacities}\label{sec:definitions}

In classical information theory, the more capable and less noisy orders play an important role~\cite{shannon1958note}. These are generally defined based on an entropic condition on the output states of a channel required to hold for a specified set of inputs. 
There are different ways to translate these classical concepts to the quantum setting; here we focus on the regularized more capable and less noisy orders introduced by Watanabe~\cite{watanabe2012private}.
In these orders, the second channel is fixed to be the complementary channel of the first one, which then leads to a characterization of the channel's capacities in terms of the capacities of the complementary channel.
A similar idea also underlies the so-called approximate degradability introduced by Sutter et.~al~\cite{sutter2017approximate}, which we discuss in more detail in Sec.~\ref{sec:approximate-degradability}.

In this work, we consider the approximate partial orders summarized in Table~\ref{tab:orders}, which were recently introduced in \cite{HRS20}.
Generally speaking, the less noisy orders are based on mixed states, i.e., either the set of mixed quantum states $\rho_{AA'}$ or classical-quantum states
\begin{align}
	\rho_{UA}=\sum_u p(u) |u\rangle\langle u| \otimes \rho^u_A, \label{EQ:rUA}
\end{align}
where each $\rho_A^u$ is a mixed state.
In contrast, the more capable orders are based on pure states, i.e., either the set of pure states $\Psi_{AA'}$ or classical-quantum states
\begin{align}
	\rho_{XA}=\sum_u p(x) |x\rangle\langle x| \otimes \Psi^x_A, \label{EQ:rXA}
\end{align}
where each $\Psi_A^x$ is a pure state.

\begin{table}[t]
	\begin{center}
		\begin{tabular}{l@{\hskip 4\tabcolsep}l@{\hskip 4\tabcolsep}l}
			\toprule
			Name & Entropic formulation & Capacity formulation \\
			\midrule
			$\epsilon$-less noisy & $I(U:E) \leq I(U:B) + \epsilon\quad\forall\rho_{UA}$ & $P^{(1)}(\cN^c) \leq \epsilon$ \\
			$\epsilon$-regularized less noisy & $I(U:E^n) \leq I(U:B^n) + n\epsilon\quad\forall\rho_{UA^n}$ & $P(\cN^c) \leq \epsilon$ \\
			$\epsilon$-fully quantum less noisy & $I(A:E) \leq I(A:B) + \epsilon\quad\forall\rho_{AA'}$ & $P_E(\cN^c) \leq \epsilon$ \\
			$\epsilon$-more capable & $I(X:E) \leq I(X:B) + \epsilon\quad\forall\rho_{XA}$ & $Q^{(1)}(\cN^c) \leq \epsilon$ \\
			$\epsilon$-regularized more capable & $I(X:E^n) \leq I(X:B^n) + n\epsilon\quad\forall\rho_{XA^n}$ & $Q(\cN^c) \leq \epsilon$ \\
			$\epsilon$-fully quantum more capable & $I(A:E) \leq I(A:B) + \epsilon\quad\forall\Psi_{AA'}$ & $Q_E(\cN^c) \leq \epsilon$ \\
			\bottomrule
	\end{tabular}\end{center}
	\caption{\label{tab:orders} Approximate partial orders discussed in this paper. For definitions of the relevant entropic quantities and capacity quantities, see Sec.~\ref{sec:definitions}. The entropic formulations are easily generalized to an arbitrary pair of channels $\cN$ and $\cM$. } 
\end{table}

We will now discuss the capacities of a quantum channel used to formulate our main results.
We focus on entropic formulas for these capacities, and refer to \cite{wilde2013quantum} for detailed operational definitions.

The classical capacity $C(\cN)$ of a quantum channel $\cN$ characterizes the optimal rate of faithful classical information transmission through the channel.
It can be expressed as~\cite{holevo1998capacity,schumacher1997sending}
\begin{align}
C(\cN) &= \lim_{n\rightarrow\infty} \frac1n \chi(\cN^{\otimes n}), \label{eq:C}\\
\chi(\cN) &=  \sup_{\rho_{XA}} I(X:B),\label{eq:holevo} 
\end{align}
where the optimization in \eqref{eq:holevo} is over classical-quantum states $\rho_{XA}$ as defined in Equation~\eqref{EQ:rXA}, which uses the fact that the optimization can be restricted to pure state ensembles \cite{wilde2013quantum}. The mutual information is evaluated on the state $(\id_X\otimes \cN)(\rho_{XA})$.
The quantity $\chi(\cN)$ is called the Holevo quantity.

The entanglement-assisted classical capacity $C_E(\cN)$ is defined as the optimal rate of faithful classical information transmission assisted by unlimited entanglement, and can be expressed as~\cite{bennett1999entanglement} 
\begin{align}
C_E(\cN) &=  \sup_{\Psi_{AA'}} I(A:B),\label{eq:CE}
\end{align}
where the mutual information is evaluated on the state $(\id_A\otimes \cN)(\Psi_{AA'})$.
A significant difference between the formulas \eqref{eq:C} and \eqref{eq:CE} is that the former is regularized (or multi-letter), referring to the limit $n\to\infty$. 
This regularization is necessary since the Holevo information is ``superadditive'': there are channels $\cN$ such that $\chi(\cN^{\otimes n}) > n\chi(\cN)$ for some $n\in\mathbb{N}$ \cite{hastings2009superadditivity}, which renders the classical capacity $C(\cN)$ intractable to compute in general.
Most capacity formulas for quantum channels suffer from this regularization problem, which creates the need for methods to bound these capacities; this is a main motivation for the present work as well.
In contrast to the formula for the classical capacity, the expression \eqref{eq:CE} for the entanglement-assisted classical capacity is a so-called single-letter formula that can be computed efficiently \cite{wilde2013quantum,fawzi2018efficient}.

The quantum capacity $Q(\cN)$ characterizes the optimal rate of faithful quantum information transmission through a channel.
It can be expressed as~\cite{BarnumNielsenEA98,BarnumKnillEA00,lloyd1997capacity,shor2002quantum,devetak2005private}
\begin{align}
Q(\cN) &= \lim_{n\rightarrow\infty} \frac1n Q^{(1)}(\cN^{\otimes n}), \label{eq:Q}\\
Q^{(1)}(\cN) &=  \sup_{\rho_{XA}} \lbrace I(X:B) - I(X:E) \rbrace \label{eq:coherent-information}\\
 &= \sup_{\Psi_{AA'}} I(A\rangle B),\label{eq:coherent-information-compact}
\end{align}
where the optimization in \eqref{eq:coherent-information} is over classical-quantum states of the form \eqref{EQ:rXA} with pure ensemble states, and the mutual informations $I(X:B)$ and $I(X:E)$ are evaluated on the states $(\id_X\otimes \cN)(\rho_{XA})$ and $(\id_X\otimes \cN^c)(\rho_{XA})$, respectively.
The alternative expression in \eqref{eq:coherent-information-compact} uses the coherent information $I(A\rangle B) = H(B) - H(AB)$.

The private capacity $P(\cN)$ characterizes the optimal rate of faithful private information transmission through a quantum channel, and can be expressed as~\cite{cai2004quantum,devetak2005private}
\begin{align}
P(\cN) &= \lim_{n\rightarrow\infty} \frac1n P^{(1)}(\cN^{\otimes n}), \label{eq:P} \\
P^{(1)}(\cN) &=  \sup_{\rho_{UA}} \lbrace I(U:B) - I(U:E) \rbrace, \label{eq:private-information}
\end{align}
where the optimization in the last line is over classical-quantum states as in Equation~\eqref{EQ:rUA} with mixed ensemble states, and the mutual informations $I(U:B)$ and $I(U:E)$ are evaluated on the states $(\id_U\otimes \cN)(\rho_{UA})$ and $(\id_U\otimes \cN^c)(\rho_{UA})$, respectively.

Similar to the classical capacity above, both the regularizations in the quantum capacity formula \eqref{eq:Q} and in the private capacity formula \eqref{eq:P} are necessary as well because of superadditivity of the underlying information quantities $Q^{(1)}(\cdot)$ and $P^{(1)}(\cdot)$~\cite{SS96,DSS98,FernWhaley08,SmithSmolin07,SRS08,LeditzkyLeungEA18,BauschLeditzky20,BauschLeditzky19,SiddhuGriffiths20,Siddhu20,Siddhu21}.
The coding theorems for the quantum and private capacity state that the single-letter information quantities are achievable lower bounds on the true capacities: For every quantum channel $\cN$,
\begin{align}
	Q^{(1)}(\cN) &\leq Q(\cN) & P^{(1)}(\cN) &\leq P(\cN).
\end{align}
However, because of the superadditivity results mentioned above it is not at all clear how large the gap in these inequalities can be.
Our capacity bounds in Cor.~\ref{Cor:CapBounds} imply that this gap is controlled by the corresponding capacity of the complementary channel, generalizing the results by Watanabe \cite{watanabe2012private} which in turn extended prior additivity results for degradable and antidegradable channels \cite{devetak2005capacity,smith2008private}.

Operationally, faithful quantum information transmission is necessarily private, and faithful private information transmission is a particular form of faithful classical information transmission.
These observations translate to the following capacity inequalities valid for any quantum channel $\cN$:
\begin{align}
	Q(\cN) \leq P(\cN) \leq C(\cN).
\end{align}
In this work we are particularly interested in the first inequality, and whether there is a gap between the quantum and private capacity.
Only a few channels with a strict separation between $Q$ and $P$ are known, among them the Horodecki channel \cite{HorodeckiHorodeckiEA05,horodecki2008low,HorodeckiHorodeckiEA09,OzolsSmithSmolin13}, the `half-rocket' channel \cite{LeungLiEA14}, and the recently introduced `platypus' channel \cite{Siddhu21,leditzky2022platypus,leditzky2022generic}.
On the other hand, Watanabe~\cite{watanabe2012private} gave sufficient criteria implying $Q(\cN) = P(\cN)$, which was previously known for degradable channels~\cite{smith2008private} and antidegradable channels (for which both capacities vanish).
One of our main results (Cor.~\ref{Cor:CapBounds}) gives a quantitative bound on the separation between $Q$ and $P$ that generalizes the result of \cite{watanabe2012private}.

Finally, we introduce two additional quantities: the entanglement-assisted private information~\cite{qi2018entanglement} 
\begin{align}
P_E(\cN) &=  \sup_{\rho_{AA'}} \lbrace I(A:B) - I(A:E) \rbrace,\label{eq:EA-priv}
\end{align}
where the mutual informations $I(A:B)$ and $I(A:E)$ are evaluated on the states $(\id_A\otimes \cN)(\rho_{AA'})$ and $(\id_A\otimes \cN^c)(\rho_{AA'})$, respectively, and its restriction to pure states,
\begin{align}
Q_E(\cN) &=  \sup_{\Psi_{AA'}} \lbrace I(A:B) - I(A:E) \rbrace. 
\end{align}
It was shown in~\cite{qi2018entanglement} that for degradable channels $P_E(\cN) = Q_E(\cN)$, and then both correspond to the entanglement-assisted private capacity of the degraded channel $\cN$. We will further expand on this comment at the end of this section. Also, if one desires an upper bound that has an operational interpretation for all channels, observe that
\begin{align} 
P^{(1)}(\cN) \leq P_E(\cN) \leq 2 Q_{ss}(\cN),\label{eq:P_E-Qss}
\end{align}
where $Q_{ss}(\cdot)$ is the quantum capacity with symmetric side channel assistance \cite{smith2008quantum}.
It can be defined as $Q_{ss}(\cN) = \sup_{d} Q^{(1)}(\cN\otimes \cA_d)$, where $\cA_d$ is a symmetric channel with $d(d+1)/2$-dimensional input and $d$-dimensional output, and zero quantum capacity by itself, $Q(\cA_d) = 0$ for all $d$. 
We discuss $Q_{ss}(\cdot)$ in more detail in Sec.~\ref{Sec:sidechannel}.
A proof of the second inequality in \eqref{eq:P_E-Qss} is provided in \Cref{sec:Qss-app}.

Before we start exploring the desired capacity bounds, we make a useful observation regarding the fully quantum more capable order and its associated capacity formula.
\begin{lemma}\label{Lem:QEQ1}
For a quantum channel $\cN$, we have
\begin{align}
Q_E(\cN) = 2 Q^{(1)}(\cN),
\end{align}
and therefore,
\begin{align}
\text{$\cN$ is $\epsilon$-more capable} \Leftrightarrow \text{$\cN$ is $2\epsilon$-fully quantum more capable.}
\end{align}
\end{lemma}
\begin{proof}
Let $V\colon \cH_A\to\cH_B\otimes \cH_E$ be a Stinespring isometry of the channel $\cN\colon A\to B$. 
For an arbitrary pure state $\Psi_{AA'}$ and $\Psi_{ABE}=V \Psi_{AA'}V^\dagger$, we have
\begin{align}
I(A\rangle B) &= H(B) - H(AB) = \frac12 \left(H(B) + H(AE) - H(AB) - H(E)\right) = \frac12 \left( I(A:B) -  I(A:E)\right). 
\end{align}
This holds for every pure state $\Psi_{AA'}$, and hence proves the first statement. The second statement then follows by definition of the orders. 
\end{proof}

\subsection{Capacity bounds} \label{sec:bounds}

We start by discussing the classical capacities of a quantum channel as a warm-up.
\begin{theorem}\label{Thm:ClassicalBounds}
For a quantum channel $\cN$, we have
\begin{alignat}{2}
Q^{(1)}(\cN) &\leq \chi(\cN)  &&\leq Q^{(1)}(\cN) + \chi(\cN^c) \label{eq:holevo-coh-bound}\\
Q(\cN) &\leq C(\cN)  &&\leq Q(\cN) + C(\cN^c) \label{eq:C-Q-bound}\\
2Q^{(1)}(\cN) &\leq C_E(\cN) &&\leq Q_E(\cN) + C_E(\cN^c)  = 2Q^{(1)}(\cN) + C_E(\cN^c)
\end{alignat}
\end{theorem}
\begin{proof} 
For each statement the first inequality is well known and is meant for comparison. The second inequality in the first statement follows by picking the optimal classical-quantum state $\rho_{XA}$ (defined in terms of a pure-state ensemble) for $\chi(\cN)$, and noting that
\begin{align}
\chi(\cN) &= I(X:B) \\
&= I(X:B) - I(X:E) + I(X:E) \\
&\leq Q^{(1)}(\cN) + \chi(\cN^c),
\end{align}
where the entropies are evaluated on the state $(I_X\otimes U_\cN)\rho_{XA}(I_X\otimes U_\cN)^\dagger$, with $U_\cN\colon \ch_A\to \cH_B\otimes \cH_E$ a Stinespring isometry for $\cN$.
The second statement follows from the first by regularizing. The third statement follows similarly to the first, using Lemma~\ref{Lem:QEQ1} for the last equality. 
\end{proof}
To make the connection to partial orders, one can note the following as direct consequences: If a channel $\cN$ is anti-more capable, we immediately have
\begin{align}
\chi(\cN)  &\leq \chi(\cN^c) \\
C_E(\cN) &\leq C_E(\cN^c).
\end{align}
Similarly, if $\cN$ is anti-regularized more capable,
\begin{align}
C(\cN)  \leq C(\cN^c).
\end{align}
Thm.~\ref{Thm:ClassicalBounds} gives our first simple bounds, and exemplifies the intuition that capacities are limited by the usefulness of the channel's complement. 

We will now consider the more interesting case of quantum capacities of quantum channels. In~\cite{HRS20} the approximate partial orders defined at the beginning of the section were used to proof a quantitative version of the previous results by Watanabe~\cite{watanabe2012private}. Those results will serve as starting point. 
\begin{theorem}[\cite{HRS20}]\label{theoboundscapacities}
Let $\cn$ be a quantum channel.
\begin{itemize}
	\item[(i)] If $\cn$ is $\epsilon$-more capable, then $Q^{(1)}(\cn) \leq P^{(1)}(\cn) \leq Q^{(1)}(\cn) + \epsilon$.
	\item [(ii)]	If $\cn$ is $\epsilon$-regularized more capable, then $	Q(\cn) \leq P(\cn) \leq Q(\cn) + \epsilon$.
	\item [(iii)]	If $\cn$ is $\epsilon$-fully quantum less noisy, then 
	$Q^{(1)}(\cn) \leq Q(\cn) \leq Q^{(1)}(\cn) + \epsilon.$
	\item [(iv)]	If $\cn$ is $\epsilon$-fully quantum less noisy and $\epsilon$-regularized more capable, then 
	$ P^{(1)}(\cn) \leq P(\cn) \leq P^{(1)}(\cn) + 2\epsilon.$
\end{itemize}
\end{theorem} 
Here, we record the following simple but important observation: for any quantum channel $\cN$, approximate partial orders can always be satisfied when considering the approximation parameters in terms of capacities of the complementary channel $\cN^c$.  For example, every channel is $\epsilon$-regularized more capable if we choose $\epsilon = Q(\cN^c)$, and similarly for the other orders. This immediately leads us to the following result.
\begin{corollary}\label{Cor:CapBounds}
For a quantum channel $\cn$, we have
\begin{alignat}{3}
Q^{(1)}(\cn) &\leq P^{(1)}(\cn) &&\leq Q^{(1)}(\cn) &&+ Q^{(1)}(\cn^c), \label{Eq:P1Q1} \\
Q(\cn) &\leq P(\cn) &&\leq Q(\cn) &&+ Q(\cn^c), \label{Eq:PQ} \\
Q^{(1)}(\cn) &\leq Q(\cn) &&\leq Q^{(1)}(\cn) &&+ P_E(\cN^c),  \label{Eq:QQPE} \\
P^{(1)}(\cn) &\leq P(\cn) &&\leq P^{(1)}(\cn) &&+ Q(\cN^c) + P_E(\cN^c).
\end{alignat}
\end{corollary}
\begin{proof}
They are a direct consequence of Theorem~\ref{theoboundscapacities} and the previous observations.
An alternative proof of eqs.~\eqref{Eq:P1Q1} and \eqref{Eq:PQ} is given in \Cref{sec:P-bound-app}, and of the inequality \eqref{eq:Qss-bound} in \Cref{sec:Qss-app}.
\end{proof}
The corollary gives operationally meaningful bounds on the maximal difference between the private and the quantum capacity and the possible advantage to be gained from regularizing the information quantities $Q^{(1)}$ and $P^{(1)}$. 

Although we are mostly concerned with upper bounds in this work, we mention here that a similar idea can also be used to detect differences between the capacities. To this end, Watanabe~\cite{watanabe2012private} proved that a channel being more capable is often also a necessary condition for the private information and the coherent information of a channel to be equal. The following is essentially~\cite[Proposition 2]{watanabe2012private} restated in the language of this work. 
\begin{corollary} 
Let $\rho^*_A$ be the optimal state achieving $Q^{(1)}(\cn)$. If $\rho^*_A$ is full rank and $Q^{(1)}(\cn^c)>0$, then
\begin{align}
P^{(1)}(\cn) > Q^{(1)}(\cn). 
\end{align}
If $|A|=2$ and $P^{(1)}(\cn) >0$, then $Q^{(1)}(\cn^c)=0$ if and only if 
\begin{align}
P^{(1)}(\cn) = Q^{(1)}(\cn). 
\end{align}
\end{corollary}
\begin{proof}
If $Q^{(1)}(\cn^c)>0$ then there exists at least one state $\rho$ for which $I(A\rangle E)_\rho >0$, and equivalently $I(A\rangle B)_\rho <0$. In this case the conditions for~\cite[Proposition 2]{watanabe2012private} are fulfilled, which proves the first statement. The second statement is a direct translation of the second part of~\cite[Proposition 2]{watanabe2012private}. 
\end{proof}

It was furthermore shown in~\cite{watanabe2012private} that a channel being less noisy is equivalent to concavity of the channel's coherent information. We now give an approximate version of this observation leading to ``approximate'' concavity and convexity results for general quantum channels. 
\begin{lemma}
For quantum states $\rho_A^i$  and a probability distribution $p(i)$, we define $\rho_A=\sum_i p(i) \rho_A^i$. A channel $\cN$ being $\epsilon$-approximate less noisy is equivalent to the statement 
\begin{align}
\sum_i p(i) I(A\rangle B)_{\rho_i} \leq I(A\rangle B)_{\rho} + \epsilon,
\end{align}
where $I(A\rangle B)_{\rho}$ is evaluated on the state $\cN(\Psi_{AA'})$ with $\Psi_{AA'}$ a purification of $\rho_A$. Similarly, a channel $\cN$ being $\epsilon$-approximate anti-less noisy is equivalent to
\begin{align}
\sum_i p(i) I(A\rangle B)_{\rho_i} \geq I(A\rangle B)_{\rho} - \epsilon.
\end{align}
For an arbitrary quantum channel $\cN$, we have
\begin{align}
I(A\rangle B)_{\rho} - P^{(1)}(\cN) \leq \sum_i p(i) I(A\rangle B)_{\rho_i} \leq I(A\rangle B)_{\rho} + P^{(1)}(\cN^c). 
\end{align}
\end{lemma}
\begin{proof}
The first and second statement follow by adjusting the proof of~\cite[Proposition 3]{watanabe2012private} using the approximate order. The third statement is then a direct consequence of the fact that every channel is $\epsilon$-approximate less noisy with $\epsilon=P^{(1)}(\cN^c)$ and $\epsilon$-approximate anti-less noisy with $\epsilon=P^{(1)}(\cN)$. 
\end{proof}

While we use the fully quantum less noisy order in Theorem~\ref{theoboundscapacities}, it was shown in \cite{watanabe2012private} that for $\epsilon=0$ the same can be proved using the regularized less noisy order. To this end, the author takes a detour using an alternative partial order based on the quantum relative entropy
\begin{align}
	D(\rho\|\sigma) = \begin{cases}
		\tr[\rho(\log\rho - \log\sigma)] & \text{if $\supp\rho \subset \supp\sigma$,}\\
		\infty & \text{otherwise,}
	\end{cases}
\end{align}
where $\supp X \coloneqq \im (\lim_{\alpha\to 0} X^\alpha)$ denotes the support of an operator $X$.
We will for now define the following auxiliary quantities. 
\begin{definition}
For a quantum channel $\cN$ we define the following quantities,
\begin{align}
R^{(1)}(\cN) &= \sup_{\rho_A, \sigma_A} D(\cN(\rho) \| \cN(\sigma)) -  D(\cN^c(\rho) \| \cN^c(\sigma)) , \\
R(\cN) &= \lim_{n\rightarrow\infty} \frac1n R^{(1)}(\cN^{\otimes n}) .
\end{align}
\end{definition}
Going back to the work of~\cite{watanabe2012private}, one can find the following inequality by adjusting their proof, 
\begin{align}
Q^{(1)}(\cn) \leq Q(\cn) \leq Q^{(1)}(\cn) + R(\cN^c). 
\end{align}
The quantity $R(\cN^c)$ is interesting here, because it was shown in~\cite{watanabe2012private} that the condition $R(\cN^c)=0$ is equivalent to $P(\cN^c)=0$, i.e., the partial orders induced by the two quantities are the same. Unfortunately, the same does not hold true for values other than $0$, i.e., the approximate partial orders; see Appendix~\ref{Ap:Relations} for an example.
In summary we can prove the following result.
\begin{theorem}\label{thm:Q-Qss}
For a quantum channel $\cN$ we have
\begin{alignat}{3}
Q^{(1)}(\cn) &\leq Q(\cn) &&\leq Q^{(1)}(\cn) &&+ M(\cn^c) \label{eq:Qss-bound}\\ 
P^{(1)}(\cn) &\leq P(\cn) &&\leq P^{(1)}(\cn) &&+ Q(\cn^c) + M(\cn^c),
\end{alignat}
where $M(\cN^c) = \min\{ R(\cN^c), P_E(\cN^c)\}$. 
\end{theorem}
Motivated by the above discussion, we make the following conjecture. 
\begin{conjecture}
For a quantum channel $\cN$ we have
\begin{alignat}{3}
Q^{(1)}(\cn) &\leq Q(\cn) &&\leq Q^{(1)}(\cn) &&+ P(\cn^c) \\ 
P^{(1)}(\cn) &\leq P(\cn) &&\leq P^{(1)}(\cn) &&+ Q(\cn^c) + P(\cn^c) \leq P^{(1)}(\cn) + 2P(\cn^c).
\end{alignat}
\end{conjecture}

Finally, we derive bounds that relate the quantities $P_E$ and $Q_E$. 
\begin{theorem}
Let $\cN$ be an $\epsilon$-fully quantum less noisy quantum channel. Then,
\begin{align}
Q_E(\cN) \leq P_E(\cN) \leq Q_E(\cN) + \epsilon. \label{Eq:CEPE}
\end{align}
Therefore, we have for any quantum channel $\cN$ that
\begin{align}
Q_E(\cN) \leq P_E(\cN) \leq Q_E(\cN) + P_E(\cN^c),
\end{align}
and equivalently,
\begin{align}
Q^{(1)}(\cN) \leq \frac12 P_E(\cN) \leq Q^{(1)}(\cN) + \frac12 P_E(\cN^c).
\end{align}
\end{theorem} 
\begin{proof} 
The first inequality in Equation~\eqref{Eq:CEPE} follows be definition of the two quantities. We now prove the second inequality. Let $\rho_{AA'}$ be the state that achieves the optimal value of $P_E(\cN)$, and let $\Psi_{ABER}=\cN(\Psi_{AA'R})$ where $\Psi_{AA'R}$ is a purification of $\rho_{AA'}$. Observe the following,
\begin{align}
&I(A:B) - I(A:E) \\
&= H(B) - H(AB) - H(E) + H(AE) \\
&= H(B) - H(ER) - H(E) + H(BR) - H(RA) + H(RA) - H(RAB) + H(E) + H(RAE) - H(B)  \\
&= I(RA:B) - I(RA:E) + I(R:E) - I(R:B) \\
&\leq Q_E(\cN) + \epsilon, 
\end{align}
where the second equality makes several uses of the purity of $\Psi_{ABER}$. The inequality follows since the system $AR$ purifies the channel input, and because $\cN$ is $\epsilon$-fully quantum less noisy. 

The next statement of the theorem follows because every channel is $\epsilon$-fully quantum less noisy with $\epsilon=P_E(\cN^c)$. The last statement follows from Lemma~\ref{Lem:QEQ1}. 
\end{proof}
A special case of the above is that, if $\cN$ is fully quantum less noisy, then
\begin{align}
P_E(\cN) = Q_E(\cN) = 2Q(\cN).
\end{align} 
This should be compared to~\cite{qi2018entanglement} where the first equality was shown for the potentially smaller set of degradable channels, but on the other hand in the more general setting of broadcast channels.

\subsection{Comparison to approximate degradability bounds}\label{sec:approximate-degradability}
A particularly useful order for bounding channel capacities is $\epsilon$-degradabiltiy introduced by Sutter et al.~\cite{sutter2017approximate}. 
To define it, we consider the diamond norm of a superoperator $\Phi\colon \cB(\cH_1)\to\cB(\cH_2)$ between the algebras $\cB(\cH_i)$ of linear operators on Hilbert spaces $\cH_1$ and $\cH_2$,
\begin{align}
	\|\Phi\|_\diamond \coloneqq \sup\lbrace \|(\id_{\cH_1}\otimes \Phi)(X)\|_1\colon X\in \cB(\cH_1\otimes \cH_1), \|X\|_1 = 1\rbrace,
\end{align}
where $\|X\|_1 = \tr \sqrt{X^\dagger X}$ denotes the trace norm.
Both the diamond norm and the trace norm can be computed efficiently using semidefinite programming \cite{Wat09,watrous2018theory}.
\begin{definition}[\cite{sutter2017approximate}]\label{def:eps-degraded}
	A channel $\cn$ is said to be an $\epsilon$-degraded version of $\cm$ if there exists a channel $\Theta$ such that $\| \cn - \Theta\circ\cm \|_\diamond \leq \epsilon$. 
\end{definition}
We will show in this section that approximate degradability implies the orders considered in Sec.~\ref{sec:bounds}; as a consequence, our bounds are at least as good as the ones derived in \cite{sutter2017approximate}, and can sometime improve upon them.
A similar discussion can be found in~\cite{HRS20}, but we will add a few new elements leading to slightly improved constants and a simpler derivation. 
We start by defining the following two functions,
\begin{align}
f_1(|E|, \epsilon) &= \frac{\epsilon}{2}\log(|E|-1) + h\left(\frac{\epsilon}{2}\right), \\
f_2(|E|, \epsilon) &= \epsilon\log|E| +\left(1+\frac{\epsilon}{2}\right)  h\left(\frac{\epsilon}{2+\epsilon}\right),
\end{align}
where $h(x)$ is the binary entropy. Note that $f_1(|E|, \epsilon) \leq f_2(|E|, \epsilon)$. As a special case of Def.~\ref{def:eps-degraded}, a channel $\cN$ is called $\epsilon$-degradable if there exists another channel $\cD$ such that
\begin{align}
\| \cN^c - \cD\circ\cN \|_\diamond \leq \epsilon. 
\end{align}
The main results of~\cite{sutter2017approximate} can be stated as follows. 
\begin{theorem}[Theorem 3.4 in~\cite{sutter2017approximate}]\label{Thm:Sutter}
Let $\cN$ be $\epsilon$-degradable, then
\begin{align}
Q^{(1)}(\cN) \leq Q(\cN) &\leq Q^{(1)}(\cN) + f_1(|E|, \epsilon) + f_2(|E|, \epsilon), \\
P^{(1)}(\cN) \leq P(\cN) &\leq P^{(1)}(\cN) + f_1(|E|, \epsilon) + 3 f_2(|E|, \epsilon), \\
Q^{(1)}(\cN) \leq P^{(1)}(\cN) &\leq Q^{(1)}(\cN) + f_1(|E|, \epsilon) + f_2(|E|, \epsilon).  
\end{align}
\end{theorem}
The proofs of these results are reminiscent of \cite{LS08}, and use continuity bounds from~\cite{audenaert2007sharp,AFW} as the main tool.
We state these continuity bounds here, adding a third one for classical-quantum states recently proved in~\cite{wilde2020optimal}. 
For two states $\rho$ and $\sigma$ with $\frac12 \|\rho-\sigma\|_1 \leq\epsilon$, it holds that
\begin{align}
| H(A)_\rho - H(A)_\sigma | &\leq f_1(|A|,2\epsilon), \\
| H(A|X)_\rho -H(A|X)_\sigma | &\leq f_1(|A|,2\epsilon), \\
| H(A|B)_\rho -H(A|B)_\sigma | &\leq f_2(|A|,2\epsilon).
\end{align}
It was also shown in~\cite{sutter2017approximate} that, if $\cN$ is $\epsilon$-anti degradable, then
\begin{align}
Q(\cN) \leq P(\cN) \leq f_1(|B|, \epsilon) + f_2(|B|, \epsilon). 
\end{align}
Similarly, we can easily see the following. 
\begin{lemma}
If $\cN$ is $\epsilon$-anti degradable, then
\begin{align}
P_E(\cN) \leq  f_1(|B|, \epsilon) + f_2(|B|, \epsilon), \\
Q^{(1)}(\cN) \leq P^{(1)}(\cN) \leq 2 f_1(|B|, \epsilon). 
\end{align}
\end{lemma}
\begin{proof}
This follows directly from data-processing and the continuity bounds mentioned above. 
\end{proof}
Combining this with our new capacity bounds, we obtain the following result. 
\begin{corollary}
If $\cN$ is $\epsilon$-degradable, then
\begin{alignat}{2}
Q^{(1)}(\cN) &\leq Q(\cN) &&\leq Q^{(1)}(\cN) + f_1(|E|, \epsilon) + f_2(|E|, \epsilon), \\
P^{(1)}(\cN) &\leq P(\cN) &&\leq P^{(1)}(\cN) + 2 f_1(|E|, \epsilon) + 2 f_2(|E|, \epsilon), \\
Q^{(1)}(\cN) &\leq P^{(1)}(\cN) &&\leq Q^{(1)}(\cN) + 2f_1(|E|, \epsilon), \\
Q(\cN) &\leq P(\cN) &&\leq Q(\cN) + f_1(|E|, \epsilon) + f_2(|E|, \epsilon). 
\end{alignat}
\end{corollary} 
As the bounds in Theorem~\ref{Thm:Sutter} are also primarily based on continuity bounds, our improvements are mostly due to the different proof technique, the fact that our bounds allow the use of the improved continuity bound from~\cite{wilde2020optimal}, and a clean regularization for the final inequality. 
The approximate degradability bounds above follow directly from further manipulating our operational capacity bounds, and hence our bounds must be at least as good as those in~\cite{sutter2017approximate}. 

\subsection{Examples}\label{sec:examples}

The bounds in this work provide new operational insights into how different capacities interact with each other. By themselves, our bounds are competitive to the best bounds in the literature; in particular, they are at least as good as those obtained from approximate degradability~\cite{sutter2017approximate}. Furthermore, our bounds can easily be combined with other capacity bounds to also give numerical bounds on the discussed capacities. 
Notably, only few bounds for the private capacity and the one-way distillable key are known (see \cite{leditzky2022platypus} for a recently defined efficiently computable bound on $P(\cdot)$ based on a result of \cite{fang2021geometric}).
Using available quantum capacity bounds, our results yield bounds on these private capacity quantities as an application.
To illustrate this technique, we consider the Horodecki channel $\cN_H$~\cite{HorodeckiHorodeckiEA05,horodecki2008low,HorodeckiHorodeckiEA09}, a so-called entanglement-binding channel \cite{HorodeckiHorodeckiEA00} satisfying $0 = Q(\cN_H) < P(\cN_H)$.
To obtain a bound on $P(\cN_H)$, we may combine our bound in \eqref{eq:QP-cap-bounds} with any available quantum capacity bound.
Using the SDP-computable bound from~\cite{wang2018semidefinite} gives
\begin{align}
	P(\cn_H) \leq Q(\cn_H) + Q(\cn_H^c) \leq 0.7284. 
\end{align}

Evidently, the bound $P(\cN)\leq Q(\cN)+Q(\cN^c)$ in \eqref{eq:QP-cap-bounds} is particularly strong when both $\cN$ and $\cN^c$ have small quantum capacity.
This holds for example when both $\cN$ and $\cN^c$ are approximately PPT channels, i.e., all of their output states are close (e.g., in trace distance) to a PPT state.
Since a channel is PPT if and only if its Choi operator is PPT, we thus look for channels $\cN$ such that the Choi operators of $\cN$ and $\cN^c$ are both close to PPT states.
We found numerical examples of such channels in low dimensions, e.g., $\dim\cH_A = 3 = \dim\cH_B$ and $\dim\cH_E=4$.
The SDP upper bound on $Q(\cdot)$ from \cite{wang2018semidefinite} detects (approximately) PPT channels, and thus yields small values for these channels.
For example, we found channels $\cN$ for which both $Q(\cN)$ and $Q(\cN^c)$ are $\approx 0.02$, implying small private capacity via our bound \eqref{eq:QP-cap-bounds}.
Interestingly, these channels still have strictly positive quantum and private capacity, as we were also able to certify strictly positive coherent information (typically $\approx 10^{-4}$) for them.

It is also interesting to consider the exact case of the above examples, i.e., quantum channels $\cN$ for which both $\cN$ and $\cN^c$ are PPT.
For such `bi-PPT' channels, our bound \eqref{eq:QP-cap-bounds} implies $P(\cN)=0$, and hence these channels would comprise a new class of zero-private-capacity channels, provided they are not also antidegradable.
Having two different classes of zero-private-capacity channels may lead to an observation of superactivation of the private capacity \cite{smith2012detecting}, which has not been found so far.
Our approximate examples from above may hint at the existence of such channels.

\section{Partial orders on quantum states}\label{Sec:States}

Recently the concept of (approximate) degradability was transferred to quantum states in~\cite{leditzky2017useful}. 
We consider a bipartite quantum state $\rho_{AB}$ with purification $\Phi_{ABE}$. The state $\rho_{AB}$ is called degradable if there exists a channel $\cD_{B\rightarrow E}$ such that 
\begin{align}
\rho_{AE} = \cD_{B\rightarrow E}(\rho_{AB}), 
\end{align}
where $\rho_{AE}=\tr_B \Phi_{ABE}$. From now on we will sometimes use the notation $\rho_{AE} =: \rho^c_{AB}$ to emphasize the role of $\rho_{AE}$ as the complementary state, in analogy to the complementary channel of a quantum channel.

It seems natural now to define new partial orders on states motivated by operational quantities. We pick the one-way distillable entanglement and secret key as our quantities of choice. Devetak and Winter showed~\cite{devetak2005distillation} that the one-way distillable secret key is given by 
\begin{align}
K_{\rightarrow}(\rho_{ABE}) = \lim_{n\rightarrow\infty}\frac1n K^{(1)}_{\rightarrow}(\rho_{ABE}^{\otimes n})
\end{align}
with
\begin{align}
K^{(1)}_{\rightarrow}(\rho_{ABE}) = \max_{Q, T|X} I(X:B|T) - I(X:E|T)
\end{align}
evaluated on
\begin{align}
\omega_{TXBE} = \sum_{t,x} R(t|x) P(x) |t\rangle\langle t|_T \otimes |x\rangle\langle x|_X \otimes \tr_A(\rho_{ABE}(Q_x\otimes\11_{BE})).
\end{align}
Here, $\lbrace Q_x\rbrace_x$ is a positive operator-valued measure (POVM), that is, $Q_x\geq 0$ for all $x$ and $\sum_x Q_x = \11_A$.
The one-way distillable entanglement is given by
\begin{align}
D_{\rightarrow}(\rho_{AB}) = \lim_{n\rightarrow\infty}\frac1n D^{(1)}_{\rightarrow}(\rho_{AB}^{\otimes n})
\end{align}
with 
\begin{align}
D^{(1)}_{\rightarrow}(\rho_{AB}) = \max_\cT I(A' \rangle BX), 
\end{align}
evaluated on $\cT_{A\rightarrow A'X}(\rho_{AB})$ where $\cT_{A\rightarrow A'X}$ is a quantum instrument. 

Generally, a quantum instrument is a map $\cT_{A\rightarrow A'X}(\cdot ) = \sum_x T_x(\cdot) \otimes |x\rl x|_X$ with each map $T_x$ being CP and such that $\sum_x T_x$ is also TP. It was shown in~\cite{devetak2005distillation} that, when considering the one-way distillable entanglement, it is sufficient to optimize over instruments where each $T_x$ is described by only one Kraus operator, i.e., $T_x(\cdot)=K_x \cdot K_x^\dagger$. Additionally, they showed that one can further restrict to the case where $K_x\geq 0$. With these observations it follows that every considered instrument is equivalently described by a POVM $\{ K_x^2 \}_x$. For the remainder of this work all instruments will be of this restricted form and this will allow us to discuss secret key and entanglement on equal footing. 

Next, for the purpose of this work we shall specify a setting that brings both quantities defined above closer together. We want to consider a state $\rho_{AB}$ with purification $\Phi_{ABE}$. When distilling secret key we give the full environment system to the eavesdropper and define $K_{\rightarrow}(\rho_{AB}) := K_{\rightarrow}(\Phi_{ABE})$. Instead of the measurement $Q$ we can optimize over an instrument, of the form as just discussed, and discard the output quantum state. We can further generalize by considering an isometric extension of the instrument $V \coloneqq V_{A\rightarrow A'X\bar X}$, as e.g.~done in~\cite{leditzky2017useful}. Taking all this together, we define the following pure quantum state:
\begin{align}
\Psi_{X\bar X A'BE} = V \Phi_{ABE} V^\dagger = \sum_{x,y}  \sqrt{P(x)P(y)} |x\rangle\langle y|_X \otimes |x\rangle\langle y|_{\bar X} \otimes K_x \Phi_{ABE} K_y^\dagger, 
\end{align}
for which
\begin{align}
\tr_{\bar X} \Psi_{X\bar X A'BE}  = \cT_{A\rightarrow A'X}(\Phi_{ABE}) =  \sum_{x}  P(x) |x\rl x|_X \otimes K_x\Phi_{ABE}K_x^\dagger,
\end{align}
is exactly the state we optimize over in the distillable entanglement. 
Now, applying a classical channel $R\colon X\to T$ to the system $X$, we get the following state,
\begin{align}
\omega_{TXA'BE} = \sum_{t,x} R(t|x) P(x) |t\rangle\langle t|_T \otimes |x\rangle\langle x|_X \otimes T_x(\Phi_{ABE}), 
\end{align}
which is the state we optimize over when considering the distillable secret key. 
It follows that we can evaluate both $K^{(1)}_{\rightarrow}(\rho_{AB})$ and $D^{(1)}_{\rightarrow}(\rho_{AB})$ on essentially the same state. 

We now define the partial orders discussed in the sequel.
\begin{definition}\label{def:state-partial-orders}
A quantum state $\rho_{AB}$ is called:
\begin{itemize}
	\item \textit{$\epsilon$-regularized more secret} if $K_{\rightarrow}(\rho_{AB}^c) \leq \epsilon$;
	\item \textit{$\epsilon$-more secret} if $K^{(1)}_{\rightarrow}(\rho_{AB}^c) \leq \epsilon$;
	\item \textit{$\epsilon$-regularized more informative} if $D_{\rightarrow}(\rho_{AB}^c) \leq \epsilon$;
	\item \textit{$\epsilon$-more informative} if $D^{(1)}_{\rightarrow}(\rho_{AB}^c) \leq \epsilon$.
\end{itemize}
For $\epsilon=0$ we drop the $\epsilon$ in the name in each case. We define the corresponding anti-orders by exchanging $\rho_{AB}^c$ with $\rho_{AB}$. 
\end{definition} 

It is clear from \Cref{def:state-partial-orders} that e.g.~$\epsilon$-anti regularized more secret implies small distillable secret key, i.e., $K_{\rightarrow}(\rho_{AB})\leq\epsilon$, and similar for the others. 

Our next goal is to rephrase the partial orders in terms of entropic inequalities. $K^{(1)}_{\rightarrow}$ already has a convenient form for that, but it will be useful to find an alternate expression for $D^{(1)}_{\rightarrow}$. Note that we can evaluate $D^{(1)}_{\rightarrow}$ on the pure state $\Psi_{X\bar X A'BE}$ defined above. We then get
\begin{align}
I(A'\rangle BX) &= H(BX) - H(A'BX) \\
&= H(BX) - H(\bar X E)  \\
&= H(BX) - H(X E)  \\
&= H(B | X) - H(E | X) , 
\end{align}
allowing us to write 
\begin{align}
D^{(1)}_{\rightarrow}(\rho_{AB}) = \max_\cT H(B | X) - H(E | X). 
\end{align}
We are now ready to give the following equivalences. 
\begin{lemma}
The state $\rho_{AB}$ is:
\begin{itemize}
	\item \textit{$\epsilon$-regularized more secret} iff for all $n\geq 1$, classical channels $R$ and quantum instruments $\cT$ applied to $\rho_{AB}^{\otimes n}$, we have 
\begin{align}
I(X:E^n|T) \leq I(X:B^n|T)  + n\epsilon;
\end{align}
\item \textit{$\epsilon$-more secret} iff for all $R$, $\cT$ applied to $\rho_{AB}$, we have 
\begin{align}
I(X:E|T) \leq I(X:B|T)  + \epsilon;
\end{align}
\item \textit{$\epsilon$-regularized more informative} iff for all $n\geq 1$ and $\cT$  applied to $\rho_{AB}^{\otimes n}$, we have 
\begin{align}
H(E^n | X) \leq H(B^n | X)  + n\epsilon;
\end{align}
\item \textit{$\epsilon$-more informative} iff for all $\cT$ applied to $\rho_{AB}$, we have 
\begin{align}
H(E | X) \leq H(B | X)   + \epsilon.
\end{align}
\end{itemize}
\end{lemma}
\begin{proof}
Follows from the above considerations.
\end{proof}

Although not immediately obvious from the entropic formulation, we have $D^{(1)}_{\rightarrow}(\rho_{AB}) \leq K^{(1)}_{\rightarrow}(\rho_{AB})$. 
Hence, more secret implies more informative and the same holds for the corresponding regularizations. 
Note that $\epsilon$-regularized more secret also implies the weaker condition,
\begin{align}
I(X:E^n) \leq I(X:B^n)  + n\epsilon, \label{Eq:weakerCond}
\end{align}
to hold for every instrument $\cT$ and all $n\geq 1$. This follows simply by considering the special case where the classical map $\cR$ is trivial. 

We now come to the first application. 
\begin{theorem}\label{Thm:distEntReg}
If the state $\rho_{AB}$ is $\epsilon$-regularized more secret, then we have
\begin{align}
D^{(1)}_{\rightarrow}(\rho_{AB}) \leq D_{\rightarrow}(\rho_{AB}) \leq D^{(1)}_{\rightarrow}(\rho_{AB}) + \epsilon, \label{Eq:D1bound}
\end{align}
and therefore for every state $\rho_{AB}$, 
\begin{align}
D^{(1)}_{\rightarrow}(\rho_{AB}) \leq D_{\rightarrow}(\rho_{AB}) \leq D^{(1)}_{\rightarrow}(\rho_{AB}) + K_{\rightarrow}(\rho_{AB}^c)
\end{align}
\end{theorem}
\begin{proof}
We start by proving the first claim. Note that $D^{(1)}_{\rightarrow}(\rho_{AB}) \leq D_{\rightarrow}(\rho_{AB})$ holds by definition. Next, we show that  $D^{(1)}_{\rightarrow}(\rho_{AB}) $ is approximately additive if $\rho_{AB}$ is $\epsilon$-regularized more secret:
\begin{align}
D^{(1)}_{\rightarrow}(\rho_{AB}) &= \max_\cT H(B | X)_\Psi - H(E | X)_\Psi \\
&\leq  H(B)_\Psi - H(E)_\Psi + \epsilon \\
&= H(B)_\Phi - H(E)_\Phi + \epsilon \\
&= H(B)_\Phi - H(AB)_\Phi + \epsilon \\
&= I(A\rangle B)_\Phi + \epsilon,
\end{align}
where the inequality follows from \eqref{Eq:weakerCond}, and the second equality is because $\Psi_{BE}=\Phi_{BE}$. The other steps are straightforward. Applying the same steps to $\rho_{AB}^{\otimes n}$, we get
\begin{align}
D^{(1)}_{\rightarrow}(\rho_{AB}^{\otimes n}) = I(A^n\rangle B^n)_{\Phi^{\otimes n}} + n\epsilon = n I(A\rangle B)_{\Phi} + n\epsilon , 
\end{align}
because the coherent information is additive on product states. Regularizing finishes the proof of the first statement. 
By definition of the order, every state is $\epsilon$-regularized more secret with $\epsilon=K_{\rightarrow}(\rho_{AB}^c)$, which proves the second claim.
\end{proof}

Note that Equation~\eqref{Eq:D1bound} is similar to one of the main results in~\cite{leditzky2017useful}, but with a weaker requirement on the state. We also remark that the proof of \Cref{Thm:distEntReg} does not use the full power of the more secret ordering, but only the weaker condition in Equation~\ref{Eq:weakerCond}. Hence, this condition could itself be seen as a partial order on quantum states. 

Our next results bounds the distillable secret key in terms of the distillable entanglement:
\begin{theorem}\label{Thm:distEntKey}
If the state $\rho_{AB}$ is $\epsilon$-more informative, then we have that 
\begin{align}
D^{(1)}_{\rightarrow}(\rho_{AB}) \leq K^{(1)}_{\rightarrow}(\rho_{AB}) \leq D^{(1)}_{\rightarrow}(\rho_{AB}) + \epsilon, \label{Eq:K1bound}
\end{align}
If the state $\rho_{AB}$ is $\epsilon$-regularized more informative, then we have that 
\begin{align}
D_{\rightarrow}(\rho_{AB}) \leq K_{\rightarrow}(\rho_{AB}) \leq D_{\rightarrow}(\rho_{AB}) + \epsilon. \label{Eq:K2bound}
\end{align}
Therefore, for every state $\rho_{AB}$, 
\begin{align}
D^{(1)}_{\rightarrow}(\rho_{AB}) &\leq K^{(1)}_{\rightarrow}(\rho_{AB}) \leq D^{(1)}_{\rightarrow}(\rho_{AB}) + D^{(1)}_{\rightarrow}(\rho_{AB}^c)  \\
D_{\rightarrow}(\rho_{AB}) &\leq K_{\rightarrow}(\rho_{AB}) \leq D_{\rightarrow}(\rho_{AB}) + D_{\rightarrow}(\rho_{AB}^c) 
\end{align}
\end{theorem}
\begin{proof}
We fix the measurement and channel achieving the maximum in $K^{(1)}_{\rightarrow}(\rho_{AB})$, and we let $\omega$ be the corresponding output state. Then,
\begin{align}
K^{(1)}_{\rightarrow}(\rho_{AB}) &= I(X:B|T) - I(X:E|T) \\
&= I(XT:B) - I(XT:E) + I(T:E) - I(T:B) \\
&= H(E|XT) - H(B|XT) + H(B|T) - H(E|T) \\
&= H(E|X) - H(B|X) + H(B|T) - H(E|T) \\
&\leq \epsilon + H(B|T) - H(E|T) \\
&\leq \epsilon + D^{(1)}_{\rightarrow}(\rho_{AB}), 
\end{align}
where the first three equalities are by definition of the involved quantities. The final equality is because we have $A\rightarrow X\rightarrow T$ and therefore $T$ does not provide additional information over $X$. The first inequality follows by definition of the more informative partial order. The second inequality follows because the remaining entropies are independent of $X$, and we can absorb the channel $X\rightarrow T$ into the choice of instrument. This proves the first claim. 

The second statement follows in the same way by considering $\rho_{AB}^{\otimes n}$ using the assumption that $\rho_{AB}$ is $\epsilon$-regularized more informative, followed by regularizing the resulting inequality. 

The final claim follows easily from the previous two, noticing that every state fulfills the needed condition for appropriately large $\epsilon$. 
\end{proof}

Finally, we can combine the previous results to get one more corollary. 
\begin{corollary}\label{Cor:DistKeyReg}
If the state $\rho_{AB}$ is $\epsilon$-regularized more secret, then we have that 
\begin{align}
K^{(1)}_{\rightarrow}(\rho_{AB}) \leq K_{\rightarrow}(\rho_{AB}) \leq K^{(1)}_{\rightarrow}(\rho_{AB}) + 2\epsilon, \label{Eq:D1bound-1}
\end{align}
and for every state $\rho_{AB}$, 
\begin{align}
K^{(1)}_{\rightarrow}(\rho_{AB}) \leq K_{\rightarrow}(\rho_{AB}) &\leq K^{(1)}_{\rightarrow}(\rho_{AB}) + K_{\rightarrow}(\rho_{AB}^c) + D_{\rightarrow}(\rho_{AB}^c) \\
&\leq K^{(1)}_{\rightarrow}(\rho_{AB}) + 2K_{\rightarrow}(\rho_{AB}^c).
\end{align}
\end{corollary}
\begin{proof}
This follows simply by combining the previous two theorems. 
\end{proof}

\section{Partial orders with symmetric side channel assistance}\label{Sec:sidechannel}

The works~\cite{smith2008quantum, smith2008private} discuss versions of the quantum and private capacities assisted by a symmetric side channel. Since the assistance can only help, they are naturally an upper bound on the respective capacities. The symmetric side channel by itself has zero quantum and private capacity and is both degradable and anti-degradable. The assisted capacities are particularly interesting, since they were proven  in~\cite{smith2008quantum, smith2008private} to be additive, i.e., $Q_{ss}^{(1)}(\cn) = Q_{ss}(\cn)$ and $P_{ss}^{(1)}(\cn) = P_{ss}(\cn)$. One can define \textit{side channel assisted} partial orders based on these quantities, analogous to the ones previously discussed.
Because of the additivity of the side channel assisted quantities, it is not necessary to consider regularizations. Note here that we have
\begin{align}
P_{ss}^{(1)}(\cn) = P_{ss}(\cn) \geq Q_{ss}^{(1)}(\cn) = Q_{ss}(\cn) \geq \frac12 P_{}(\cn), 
\end{align}
where the final inequality was proven in~\cite{smith2008zeroCapacity}. This implies in particular that 
\begin{align}
Q_{ss}^{(1)}(\cn^c) = 0  \quad\Rightarrow\quad  P_{}(\cn^c)=0, 
\end{align}
which could provide us with an easier (non-regularized) condition to determine whether a channel is regularized less noisy. On the other hand, we still lack an example of a channel that is regularized less noisy but not degradable. We know that the quantum capacity is superadditive \cite{smith2008zeroCapacity,li2009private,BrandaoOppenheimEA12,leditzky2022generic} and can in particular be superactivated, i.e., there exists a channel for which $Q(\cN)=0$ but $Q_{ss}(\cn)>0$~\cite{smith2008zeroCapacity}. Note that if we had a channel with $P(\cN^c)=0$, but 
\begin{align}
P_{ss}^{(1)}(\cn^c) > 0  \quad\Rightarrow\quad  \cN \text{ not degradable}, 
\end{align}
it would give us the desired example. It seems intuitive that a similar construction works more generally and there is a deeper connection between superactivation of the private capacity and such examples. We can make this more precise in the following observation.
\begin{corollary}
If the private capacity can be superactivated then degradable channels are a strict subset of regularized less noisy channels, i.e.
\begin{align}
P(\cdot) \text{ can be superactivated} \quad\Rightarrow\quad DEG \subsetneq LN_\infty. 
\end{align} 
\end{corollary} 
\begin{proof}
Let us assume that the private capacity can be superactivated, meaning there exist channels $\cN$ and $\cM$ such that $P(\cn)=P(\cm)=0$, but $P(\cn\otimes\cm)>0$. We observe that if $\cn$ and $\cm$ are anti-degradable, then so is $\cn\otimes\cm$, and we would have $P(\cn\otimes\cm)=0$. Therefore, by assumption, at least one of $\cn$ or $\cm$ has to be non-anti-degradable. However, again by assumption, their complements $\cn^c$ and $\cm^c$ are regularized less noisy. In summary, at least one of $\cn^c$ or $\cm^c$ is regularized less noisy but not degradable, which concludes the proof. 
\end{proof}
Interestingly, the question whether the private capacity can be superactivated is still open, despite significant effort to find an answer~\cite{li2009private, smith2009extensive, strelchuk2013parrondo}.
Of course, the above corollary implies that if all regularized less noisy channels are also degradable, the private capacity cannot be superactivated. 
Here we mention once more the observation from \Cref{sec:examples}: If there is a non-anti-degradable channel $\cN$ such that both $\cN$ and $\cN^c$ are PPT and hence have zero quantum capacity, then $P(\cN)=0$ by \Cref{Cor:CapBounds}.
Such `bi-PPT' channels $\cN$ would constitute a new class of zero-private-capacity channels, which together with anti-degradable channels may exhibit superactivation of private capacity.

Before ending this section, we briefly observe a result similar to the main theme of this work and state some bounds on the symmetric side channel assisted capacity. 
\begin{corollary}
Let $\cn$ be a quantum channel. We have
\begin{alignat}{3}
Q_{ss}(\cn) &\leq P_{ss}(\cn) &&\leq Q_{ss}(\cn) &&+ Q_{ss}(\cn^c) \label{Eq:PQss}. %
\end{alignat}
\end{corollary}
\begin{proof}
This follows because according to~\cite{smith2008quantum, smith2008private} we can write 
\begin{align}
Q_{ss}(\cn) &= Q_{ss}^{(1)}(\cn) =\sup_d Q_{}^{(1)}(\cn \otimes \cA_d) \\
P_{ss}(\cn) &= P_{ss}^{(1)}(\cn) =\sup_d P_{}^{(1)}(\cn \otimes \cA_d), 
\end{align}
as well as noticing that $(\cn \otimes \cA_d)^c = \cn^c \otimes \cA_d^c = \cn^c \otimes \cA_d$. Finally, combining both observations with Corollary~\ref{Cor:CapBounds} leads to the first result. 
\end{proof}

\section{Outlook and open problems}\label{Sec:Outlook}
In this work we derived operationally meaningful bounds on the capacities of quantum channels and quantum states.
It is an interesting (but hard) open problem to derive similar bounds for the distillable entanglement and distillable key under LOCC or PPT-preserving operations. Moreover, channel capacities with e.g.~two-way assistance might obey similar bounds. Overall, it would be interesting to investigate whether there is a more general framework in which such results can be proven. 

Finally, it would be interesting to find an example of the aforementioned class of channels that are `bi-PPT' (both the channel and its complement are PPT) but not antidegradable, or more generally any channel that is regularized less noisy but not degradable.
Such channels would give promising candidates to show superactivation of the private capacity.

\section*{Acknowledgments} 
We thank Mark M. Wilde for suggesting to apply our ideas to one-way distillable secret key and energy-constrained capacities, and Graeme Smith for helpful feedback.

This project has received funding from the European Union's Horizon 2020 research and innovation programme under the Marie Sklodowska-Curie Grant Agreement No. H2020-MSCA-IF-2020-101025848.

\appendix

\section{Simple entropic proofs of some capacity bounds}

In this appendix we provide simple entropic proofs of some of the capacity bounds stated in \Cref{sec:bounds}.

\subsection{Bound on private information and private capacity}\label{sec:P-bound-app}

We first prove the following inequalities from \Cref{Cor:CapBounds}, giving bounds on the private information and private capacity of a quantum channel:
\begin{align}
	\priv{\cN} &\leq \coh{\cN} + \coh{\cN_c} \label{eq:priv-inf-upper-bound}\\
	P(\cN) &\leq Q(\cN) + Q(\cN_c).	\label{eq:priv-cap-upper-bound}
\end{align}

Let $\lbrace p(u), \rho_A^u\rbrace_u$ be a mixed quantum state ensemble achieving $\priv{\cN}$, and let $W_\cN\colon \ch_A\to \cH_B\otimes \cH_E$ be a Stinespring isometry for $\cN$.
For each $u$, we consider spectral decompositions of the state $\rho_A^u$, 
\begin{align}
	\rho_A^u = \sum_v p(v|u) |\phi^{u,\,v}\rangle\langle \phi^{u,\,v}|_A.
\end{align} 
We then form the classical-quantum state
\begin{align}
	\sigma_{UVBE} = \sum_{u,\,v} p(u)\, p(v|u) \, |u\rangle\langle u|_U \otimes |v\rangle \langle v|_V \otimes W_\cN |\phi^{u,\,v}\rangle\langle \phi^{u,\,v}|_A W_\cN^\dagger,
\end{align}
and note that $\sigma_{UBE}$ is the classical-quantum state satisfying $\priv{\cN} = I(U:B)_\sigma - I(U:E)_\sigma$.
Consider now the following:
\begin{align}
	\priv{\cN} &= I(U:B)_{\sigma} - I(U:E)_{\sigma}\\
	&= I(UV:B) - I(UV:E) + I(V:E|U) - I(V:B|U)\\
	&= I(UV:B) - I(UV:E) + \sum_u p(u) \left[ I(V:E|U=u) - I(V:B|U=u) \right]\\
	&\leq \coh{\cN} + \sum_u p(u) \coh{\cN_c}\\
	&= \coh{\cN} + \coh{\cN_c}.
\end{align}
In the third equality we used the fact that $I(R:Q|X)_\tau = \sum_x p(x) I(R:Q|X=x)_\tau$, where $I(R:Q|X=x)_\tau = I(R:Q)_{\tau^x}$ for a classical-quantum state $\tau_{XRQ} = \sum_x p(x) |x\rangle\langle x|_X \otimes \tau^x_{RQ}$.
For the inequality, we used \eqref{eq:coherent-information} for each difference term.
The inequality \eqref{eq:priv-cap-upper-bound} now follows from regularizing \eqref{eq:priv-inf-upper-bound}.

\subsection{Bound on quantum capacity}\label{sec:Qss-app}

Here we prove the following bound on the quantum capacity, which was stated in Equation~\eqref{Eq:QQPE}:
\begin{align}
	Q(\cN) \leq Q^{(1)}(\cN) + P_E(\cN^c).\label{eq:Q-Qss}
\end{align}

Let $\rho_{A^n}$ be a state achieving $Q^{(1)}(\cN^{\otimes n})$, and let $U_\cN\colon \ch_A\to \cH_B\otimes \cH_E$ be a Stinespring isometry for $\cN$.
Consider now the following steps, which are inspired by \cite{LS08,sutter2017approximate,cross2017uniform}:
\begin{align}
	Q^{(1)}(\cN^{\otimes n}) &= H(B^n) - H(E^n)\\
	&= H(B_1\dots B_n) - H(E_1B_2\dots B_n)\notag\\
	&\quad {} + H(E_1 B_2\dots B_n) - H(E_1E_2 B_3 \dots B_n)\notag\\
	&\quad {} + \dots \notag\\
	&\quad {} + H(E_1 \dots E_{n-1} B_n) - H(E_1 \dots E_n). \label{eq:n-letter}
\end{align}
Using the notation $R_{i}^j = R_i R_{i+1}\dots R_{j-1} R_{j}$ for $i<j$, we have for each $i=1,\dots,n$ that
\begin{align}
	H(E_1^{i-1} B_i B_{i+1}^n) - H(E_1^{i-1} E_i B_{i+1}^n) &= H(B_i) - H(E_i) + I(E_i: E_1^{i-1}B_{i+1}^n) - I(B_i: E_1^{i-1}B_{i+1}^n)\\
	&\leq Q^{(1)}(\cN) + P_E(\cN^c).\label{eq:F-bound}
\end{align}
Combining \eqref{eq:n-letter} and \eqref{eq:F-bound} yields 
\begin{align}
	Q^{(1)}(\cN^{\otimes n}) \leq n Q^{(1)}(\cN) + n P_E(\cN^c)\!,
\end{align}
from which \eqref{eq:Q-Qss} follows by dividing by $n$ and taking the limit $n\to \infty$.

Finally, we show the following inequality stated in \eqref{eq:P_E-Qss}:
\begin{align}
	P_E(\cN) \leq 2 Q_{ss}(\cN).\label{eq:P_E-Qss-app}
\end{align}
To see this, recall the following rewriting of the symmetric side channel assisted capacity proved in \cite{smith2008quantum}:
\begin{align}
	\Qss{\cN} = \max_{\rho_{VWA}} \frac{1}{2} \left\lbrace I(V:WB)_\sigma - I(V:WE)_\sigma\right\rbrace,\label{eq:Qss-rewritten}
\end{align}
with the entropies evaluated on the state $\sigma_{VWBE}\coloneqq (I_{VW}\otimes U_\cN)\rho_{VWA}(I_{VW}\otimes U_\cN)^\dagger$.
Comparing this expression with the definition \eqref{eq:EA-priv} of the entanglement-assisted private information, we obtain \eqref{eq:P_E-Qss-app} by choosing the system $W$ to be trivial.
Combining \eqref{eq:Q-Qss} and \eqref{eq:P_E-Qss-app} further shows that
\begin{align}
	Q(\cN) \leq Q^{(1)}(\cN) + 2 Q_{ss}(\cN^c).
\end{align}

\section{On relationships between approximate partial orders}\label{Ap:Relations}

The relationships between different partial orders are a fundamental problem and have been investigated at several points in the literature. In particular, the relationships between most of the orders discussed in this work have been recently investigated in~\cite{HRS20}, however only in the case $\epsilon=0$. Focusing on approximate orders, it is worth noting that the picture becomes substantially more complicated when $\epsilon>0$. 

We first start with a simple example motivated by Section~\ref{Sec:sidechannel}. Note that we have the following implications:
\begin{align}
P_{ss}^{(1)}(\cn^c) = 0  \quad\Rightarrow\quad  Q_{ss}^{(1)}(\cn^c) = 0  \quad\Rightarrow\quad  P_{}(\cn^c)=0.
\end{align}
However, when allowing an approximation, if we go via $Q_{ss}^{(1)}(\cn^c)$, the best we can currently show is
\begin{align}
P_{ss}^{(1)}(\cn^c) \leq \epsilon \quad\Rightarrow\quad  Q_{ss}^{(1)}(\cn^c)\leq \epsilon  \quad\Rightarrow\quad  P_{}(\cn^c)\leq 2\epsilon, 
\end{align}
although generally we also have
\begin{align}
P_{ss}^{(1)}(\cn^c) \leq \epsilon \quad\Rightarrow\quad  P_{}(\cn^c)\leq \epsilon.
\end{align}
Therefore the implied exact partial orders have a simpler relationship than the more general approximate versions. 

We now discuss the main part of this section. Recall the definition of $R^{(1)}(\cN)$, which is, similar to the other quantities defined in this work, related to a partial order. For two quantum channels $\cN$ and $\cM$ we denote $\cN \erel{\epsilon} \cM$ if
\begin{align}
D(\cM(\rho) \| \cM(\sigma) ) \leq D(\cN(\rho) \| \cN(\sigma) ) + \epsilon \quad\forall \rho, \sigma \label{eq:rel-ent-partial}
\end{align}
and simply 
$\cN \rel \cM$ if $\epsilon=0$. 
An important technical result in~\cite{watanabe2012private} was the following observation,
\begin{align}
\cN \rel \cM  \quad \Leftrightarrow \quad \cN \sln \cM.
\end{align}
Following the proof in~\cite{watanabe2012private} it can be easily seen that also in the approximate case we still have
\begin{align}
\cN \erel{\epsilon} \cM  \quad \Rightarrow \quad \cN \esln{\epsilon} \cM. \label{Eq:ImpRelLN}
\end{align}
However, we will now see that the opposite direction is generally not true. For this, consider two erasure channels $\cE_1$ and $\cE_2$ with erasure probabilities $\epsilon_1$ and $\epsilon_2$, respectively. We know that for an erasure channel one has $I(A:B) = (1-\epsilon)I(A:A')$ and $D(\cE(\rho)\| \cE(\sigma)) = (1-\epsilon)D(\rho\|\sigma)$. First, consider the approximate less noisy condition, which for our example evaluates to
\begin{align}
(1-\epsilon_2) I(U:A) \leq (1-\epsilon_1) I(U:A) + \epsilon.
\end{align}
It can now easily be seen that the two channels are always $\epsilon$-approximately less noisy if $\epsilon = \max\{0, 2(\epsilon_1 - \epsilon_2) \log|A| \}$, because $I(U:A)\leq 2\log|A|$.
Next, note that the condition for the partial order based on relative entropy defined in \eqref{eq:rel-ent-partial} can be written as
\begin{align}
(1-\epsilon_2) D(\rho\|\sigma) \leq (1-\epsilon_1) D(\rho\|\sigma) + \epsilon.
\end{align}
Since the relative entropy can be arbitrarily large for suitably chosen $\rho$ and $\sigma$, there is in general no $\epsilon$ such that the above inequality always holds provided that $\epsilon_1>\epsilon_2$.
This proves that the reverse implication of Equation~\eqref{Eq:ImpRelLN} cannot hold.  Note however that this counterexample does not seem to work in the asymptotic setting, because both quantities collect prefactors of order $\epsilon^n$ that tend to $0$ for $n\rightarrow\infty$. Finally, the example can easily be specialized to the case where $\cM=\cN^c$, as the complementary channel of an erasure channel with erasure probability $\epsilon_1$ is an erasure channel with erasure probability $\epsilon_2=1-\epsilon_1$.

\section{Energy-constrained partial orders on quantum channels} 

In this section we return to the theme of quantum channels, but we add a twist by considering energy-constrained settings. We will again base our partial orders on the quantum and private capacities of a quantum channel, but this time we focus on their energy-constrained variants. Based on~\cite{wilde2018energy} we consider the following quantities: %
The energy-constrained quantum capacity
\begin{align}
Q_{H_A,E}(\cN) &= \lim_{n\rightarrow\infty} \frac1n Q^{(1)}_{H_{A^n},n E}(\cN^{\otimes n})\\
\text{with}\quad Q^{(1)}_{H_A,E}(\cN) &= \sup_{\substack{\rho_A \\ \tr H_A \rho_A \leq E}} \lbrace H(\cN(\rho_A)) - H(\cN^c(\rho_A)) \rbrace,
\end{align}
and the energy-constrained private capacity
\begin{align}
P_{H_A,E}(\cN) &= \lim_{n\rightarrow\infty} \frac1n P^{(1)}_{H_{A^n},n E}(\cN^{\otimes n})\\
\text{with}\quad P^{(1)}_{H_A,E}(\cN) &= \sup_{\substack{\rho_{UA} \\ \tr H_A \rho_A \leq E}} \lbrace I(U:B) - I(U:E)\rbrace. 
\end{align}
For the energy constraint, we define the Hamiltonian $H_{A^n}$ on $n$ copies of the input quantum system as the extension of the single system Hamiltonian $H_A$ as 
\begin{align}
H_{A^n} = H_A \otimes \11_A \otimes \dots \otimes \11_A + \dots + \11_A \otimes \dots \otimes \11_A \otimes H_A. 
\end{align}
Throughout this section we will assume that the finite output entropy condition holds, that is,
\begin{align}
\sup_{\substack{\rho_A \\ \tr H_A \rho_A \leq E}} H(\cN(\rho_A)) < \infty. 
\end{align}
It was shown in~\cite{wilde2018energy} that if this condition holds for a channel $\cN$, it also holds for the complementary channel $\cN^c$. 
We can now define the following energy-constrained partial orders.

\begin{definition}
A channel $\cn$ is called: 
	\begin{itemize}
	\item \textit{$(\epsilon,H_A,E)$-regularized less noisy} if $P_{H_A,E}(\cN^c) \leq \epsilon\;;$
	\item \textit{$(\epsilon,H_A,E)$-less noisy} if $P_{H_A,E}^{(1)}(\cN^c) \leq \epsilon\;;$
	\item \textit{$(\epsilon,H_A,E)$-regularized more capable}  if $Q_{H_A,E}(\cN^c) \leq \epsilon\;;$
	\item \textit{$(\epsilon,H_A,E)$-more capable} if $Q_{H_A,E}^{(1)}(\cN^c) \leq \epsilon\;.$
\end{itemize}
\end{definition}
Note that each of these partial orders has an equivalent definition via a mutual information-based condition similar to the unconstrained case, with the difference that the condition only needs to be checked on states satisfying the energy constraint. For the less noisy orderings, this is fairly obvious from the definition. For the more capable orderings, consider a state $\rho_A= \sum_i \lambda_i | \Psi_i\rl\Psi_i|$ and its extension $\rho_{UA} = \sum_i \lambda_i |i\rl i| \otimes | \Psi_i\rl\Psi_i|$. The energy constraint $\tr H_A \rho_A \leq E$ remains the same and can be interpreted as an average energy constraint of the ensemble $\{ \lambda_i, |\Psi_i\rangle\}$ as $\sum_i \lambda_i \tr H_A |\Psi_i \rl \Psi_i | \leq E$, similarly to the less noisy setting. 

In the unconstrained case we saw that the less noisy order is closely related to the concavity of a channel's coherent information. A careful check reveals that the same holds true if an energy constraint is to be obeyed. 
\begin{lemma}
A channel $\cN$ is $(\epsilon,H_A,E)$-approximate less noisy if and only if its channel coherent information is approximately concave for all quantum states $\rho_A^i$ and probability distributions $p(i)$ satisfying $\tr H_A \rho_A \leq E$ with $\rho_A=\sum_i p(i) \rho_A^i$:
\begin{align}
\sum_i p(i) I(A\rangle B)_{\rho_i} \leq I(A\rangle B)_{\rho} + \epsilon,
\end{align}
where $I(A\rangle B)_{\rho}$ is evaluated on the state $\cN(\Psi_{AA'})$ with $\Psi_{AA'}$ a purification of $\rho_A$. Similarly, a channel $\cN$ being $(\epsilon,H_A,E)$-approximate anti-less noisy is equivalent to
\begin{align}
\sum_i p(i) I(A\rangle B)_{\rho_i} \geq I(A\rangle B)_{\rho} - \epsilon.
\end{align}
For an arbitrary quantum channel $\cN$ and states obeying $\tr H_A \rho_A \leq E$, we have
\begin{align}
I(A\rangle B)_{\rho} - P_{H_A,E}^{(1)}(\cN) \leq \sum_i p(i) I(A\rangle B)_{\rho_i} \leq I(A\rangle B)_{\rho} + P_{H_A,E}^{(1)}(\cN^c). 
\end{align}
\end{lemma}

Now we would like to briefly discuss to what extent the bounds and results on channel capacities can be extended to the energy-constrained setting. To this end, we briefly revisit the approach in~\cite{watanabe2012private}. Take quantum states $\rho_A^u$ and a probability distribution $p(u)$, define $\rho_{UA} = \sum_u p(u) |u\rangle\langle u| \otimes \rho_A^u$ and $\rho_A=\tr_U\rho_{UA}$. A central observation to the proofs in~\cite{watanabe2012private} is that the following equality holds, 
\begin{align}
I(U:B) - I(U:E) &= I(A\rangle B)_{\rho} -  \sum_i p(i) I(A\rangle B)_{\rho_i} \\
&= I(A\rangle B)_{\rho} +  \sum_i p(i) I(A\rangle E)_{\rho_i}.
\end{align}
If we fix $\rho_{UA}$ to be the optimizing state in $P_{H_A,E}^{(1)}(\cN^c)$, we easily get 
\begin{align}
P_{H_A,E}^{(1)}(\cN) \leq Q_{H_A,E}^{(1)}(\cN)  + \sum_i p(i) I(A\rangle E)_{\rho_i}. 
\end{align}
In the unconstrained case it is now easy to bound each of the remaining coherent informations by $Q^{(1)}(\cN^c)$, which makes the average irrelevant and leads to our previously stated inequality. However, in the constrained case, we can not simply do the same using $Q_{H_A,E}^{(1)}(\cN^c)$ because the individual $\rho_i$ might not fulfill the energy constraint $\tr H_A \rho^i_A \leq E$. One might be tempted to remedy this problem by using concavity, but from the previous lemma it is clear that this doesn't seem to help for general channels. 
In~\cite{wilde2018energy} it was shown that $Q_{H_A,E}^{(1)}(\cN)$ equals both the energy-constrained quantum and private capacity of a degradable channel. How is this compatible with the above observations? To us, the most likely explanation seems to be the following.  Note that degradability is usually defined via the diamond norm and therefore considering all possible input states without constraint. Equivalently, we can prove
\begin{align}
P_{H_A,E}^{(1)}(\cN) \leq Q_{H_A,E}^{(1)}(\cN)  + Q^{(1)}(\cN^c), 
\end{align}
showing that $P_{H_A,E}^{(1)}(\cN) = Q_{H_A,E}^{(1)}(\cN)$ for all $\epsilon$-less noisy channels $\cN$, which is a significantly weaker requirement than degradability. 

One could similarly define a weaker form of degradability that obeys an energy constraint. It is an interesting question whether results like those  in~\cite{watanabe2012private} hold under this requirement. 
\begin{definition}
A channel $\cN$ is called $(\epsilon,H_A,E)$-degradable if there exists a channel $\cD$ such that
\begin{align}
\| \cN^c - \cD\circ\cN \|_\diamond^{H_A,E} \leq \epsilon, 
\end{align}
where $\| \Delta \|_\diamond^{H_A,E}$ is the energy-constrained diamond norm defined in~\cite{shirokov2018energy}. 
\end{definition}
The energy-constrained diamond norm has already found several applications, in particular for infinite dimensional systems, see e.g.~\cite{shirokov2018energy, winter2017energy}. 

Finally, we comment on single-letter upper bounds on regularized capacities. Note that, following the proof in~\cite{cross2017uniform}, one obtains
\begin{align}
I(A\rangle B^n) \leq \sum_i I(A\rangle B_i) + \sum_i [ I(V|B_i) - I(V|E_i) ]. 
\end{align}
However, while the state $\rho_{A^n}$ obeys the constraint $\tr H_{A^n}\rho_{A^n}\leq nE$, the best we can say about the individual $\rho_{A_i}$ is that they also obey $\tr H_{A}\rho_{A_i}\leq nE$. This leads us to the somewhat unsatisfying result
\begin{align}
Q_{H_A,E}^{(1)}(\cN^{\otimes n})  \leq n Q_{H_A,nE}^{(1)}(\cN)  +  n P_E^{H_A,nE}(\cN^c), 
\end{align}
where $P_E^{H_A,E}$ is the energy-constrained entanglement-assisted private information defined as 
\begin{align}
P_E^{H_A,E}(\cN) = \sup_{\substack{\rho_{AA'} \\ \tr H_A \rho_{A'} \leq E}} \lbrace I(A:B) - I(A:E)\rbrace. 
\end{align}
If one wishes to regularize, the energy constraints on the right hand side would become meaningless, resulting in the inequality
\begin{align}
Q_{H_A,E}(\cN)  \leq  Q^{(1)}(\cN)  +  P_E(\cN^c).
\end{align}
This implies once more that the desired simplifications only seem to hold if a requirement without energy-restriction holds. 
Note that the above behavior is certainly intuitive as the way energy-constrained capacities are regularized allows for strategies where a single input uses an arbitrarily high amount of energy as long as it is compensated by the other channel uses. The problem would be resolved if one considered a more restricted way of regularizing the quantities where each channel input is subject to a fixed energy constraint instead of an average energy constraint on the overall state.
This might also be a practically more relevant scenario.

\printbibliography[heading=bibintoc]
\end{document}